\documentclass[a4paper, USenglish, numberwithinsect, cleveref,
thm-restate,nolineno]{socg-lipics-v2021}

\pdfoutput=1
\hideLIPIcs

\title{On the Complexity of the Minimum-\texorpdfstring{\boldmath$(k,\rho)$}{(k,ρ)}-Shortcut Problem}

\author{Tatiana Rocha Avila}{Goethe University Frankfurt, Germany}{rochaavila@em.uni-frankfurt.de}{https://orcid.org/0009-0008-6450-2444}{Funded by DFG Research Unit ADYN under grant DFG 411362735.}
\author{Julian Christoph Brinkmann}{Goethe University Frankfurt, Germany}{J.Brinkmann@em.uni-frankfurt.de}{https://orcid.org/0009-0000-0332-4543}{}
\author{Alexander Leonhardt}{Goethe University Frankfurt, Germany}{aleonhardt@ae.cs.uni-frankfurt.de}{https://orcid.org/0009-0006-8263-6900}{Funded by the Deutsche Forschungsgemeinschaft (DFG) – ME 2088/5-2
(FOR 2975 — Algorithms, Dynamics, and Information Flow in Networks).}
\author{Conrad Schecker}{Goethe University Frankfurt, Germany}{schecker@em.uni-frankfurt.de}{https://orcid.org/0000-0002-8103-1911}{}

\authorrunning{T. R. Avila, J. C. Brinkmann, A. Leonhardt and C. Schecker}

\Copyright{Tatiana Rocha Avila, Julian Christoph Brinkmann, Alexander Leonhardt and Conrad Schecker} 



\EventEditors{John Q. Open and Joan R. Access}
\EventNoEds{2}
\EventLongTitle{42nd Conference on Very Important Topics (CVIT 2016)}
\EventShortTitle{CVIT 2016}
\EventAcronym{CVIT}
\EventYear{2016}
\EventDate{December 24--27, 2016}
\EventLocation{Little Whinging, United Kingdom}
\EventLogo{}
\SeriesVolume{42}
\ArticleNo{23}

\usepackage{ragged2e}
\usepackage{mathtools}
\usepackage{tikz}
\usetikzlibrary{fadings, 3d, positioning, calc, decorations.pathreplacing, arrows.meta}
\usepackage{hyphenat}

\newcommand{\rhoconstspt}[3]{\ensuremath{T^{#2\if\relax\detokenize{#3}\relax\else,#3\fi}_{#1}}}

\newcommand{\minkrhoShortcutProb}[1][k]{\ensuremath{\text{min-}(#1,\rho)\text{-}\mathrm{Shortcut}}}
\newcommand{\krhoShortcutProb}{\ensuremath{(k,\rho)\text{-}\mathrm{Shortcut}}}
\newcommand{\minkShortcutProb}{\ensuremath{\text{min-}k}\text{-}\mathrm{Shortcut}}

\newcommand{\krhoTiebreaker}{\ensuremath{(k,\rho)\text{-}\mathrm{Tiebreaker}}}

\newcommand{\HittingSet}{\textsc{HittingSet}}
\newcommand{\dHittingSet}[1]{\ensuremath{#1\text{-}\mathrm{HittingSet}}}

\DeclareMathOperator{\dist}{dist}
\DeclareMathOperator{\hopdist}{hopdist}

\newcommand{\N}{\mathbb{N}}
\newcommand{\krhoRestrictedSubgraph}[1]{\ensuremath{\left.G\right|_{k,#1}}}
\newcommand{\krhoBall}[1]{\ensuremath{(k,#1)\text{-}\mathrm{Ball}}}
\def\NP{\ensuremath{\mathrm{NP}}}
\def\P{\ensuremath{\mathrm{P}}}

\begin{document}

\maketitle
\begin{abstract}

We consider the \textsc{Minimum}-$(k,\rho)$\textsc{-Shortcut} problem~($\minkrhoShortcutProb$), where the goal is to find the smallest set of shortcut edges such that every vertex in a given graph can reach its~$\rho$ closest vertices using paths of at most~$k$ edges.
This is a fundamental graph optimization problem used to accelerate parallel shortest path algorithms.

It is well-known that the problem is trivially solvable for the cases $k=1$ and $k\geq\rho$.
While recent work by Leonhardt, Meyer, and Penschuck (ESA 2024) showed that in undirected graphs \(\minkrhoShortcutProb\) is \NP-hard for $k\geq 3$ if $\rho=\Theta(n^\varepsilon)$, the boundary where the problem transitions from polynomial-time solvable to \NP-hard remained open.

In this paper, we narrow this gap significantly.
We present a simpler and more direct reduction from the Hitting Set problem which establishes that \(\minkrhoShortcutProb\) is \NP-hard for~$k\geq2$ and~$\rho\geq k+2$ in both undirected and directed graphs.
Complementing this, we use the symmetry of the undirected case to show that $\rho=k+1$ is solvable in polynomial time, a regime where the directed version remains a candidate for \NP-hardness.
Therefore, we obtain an almost complete characterization of the complexity of \(\minkrhoShortcutProb\), with the sole remaining open case being \({\rho = k+1}\) in the directed setting.

\end{abstract}
\newpage

\section{Introduction}\label{sec:intro}
The design of algorithms to solve the single source shortest path problem is one of the fundamental challenges of computer science.
Since the inception of the well-known Dijkstra \cite{DBLP:books/mc/22/Dijkstra22a} and Bellman-Ford algorithms \cite{bellman, ford1956network}, efficient practical sequential algorithms have emerged~\cite{DBLP:journals/siamcomp/Goldberg08, DBLP:conf/wea/GeisbergerSSD08}.
In contrast, the progress of a parallel counterpart, namely the parallel single source shortest path problem (PSSSP) on general graphs is limited.

Traditional PSSSP approaches often struggle to achieve low depth (i.e. high parallelism) since this usually requires more work than the sequential Dijkstra's algorithm. Recently, Cao and Fineman~\cite{DBLP:conf/soda/CaoF23} presented a novel algorithm that improved the theoretical upper bounds for solving PSSSP on directed graphs.
It achieves almost linear work~$\widetilde{\mathcal{O}}(m)$ and almost square root depth~$n^{0.5+o(1)}$. 
In practice, variants based on the \textit{stepping framework} are used \cite{DBLP:journals/jal/MeyerS03,DBLP:conf/spaa/Blelloch0ST16,DBLP:conf/spaa/Dong00Z21}.
These variants exhibit superior experimental performance over Bellman-Ford achieving linear depth in specific instances, although they do not provide theoretical guarantees for general instances. 
The gap between theory and practice arises from intricate design and high constant factors involved in novel theoretical algorithms.

To support the theoretical bounds for \textit{stepping based} PSSSP algorithms, Blelloch, Gu, Sun, and Tangwongsan~\cite{DBLP:conf/spaa/Blelloch0ST16} introduced a new graph notion alongside their \textsc{Radius-Stepping} algorithm for undirected graphs.
They call a graph a \textit{$(k,\rho)$-graph}, if every vertex of the graph can reach its~$\rho$-closest vertices by using shortest paths with at most~$k$ edges.
In this setting, the algorithm's  work is almost linear in the number of edges $\Tilde{\mathcal{O}}(m)$ and its depth is bounded by $O\Big(\frac{n}{\rho}\log n\log(\rho L)\Big)$, where $L$ is the ratio between the maximum and minimum edge weight of the graph. 

Any graph can be viewed as a $(k,\rho)$-graph for an appropriate choice of the parameters $k$ and $\rho$. Furthermore, for fixed values of $k$ and $\rho$, any graph can be transformed into a $(k,\rho)$-graph by adding \emph{shortcuts}, that is, additional edges preserving all pairwise distances in the original graph.
This naturally leads to the following optimization problem:
\begin{definition}
    The \textnormal{\textsl{\textsc{Minimum}-$(k,\rho)$\textsc{-Shortcut}}} problem \emph{(\minkrhoShortcutProb)}:
    \begin{itemize}
        \item Input: A weighted graph~$G$.
        \item Problem: Find shortcut set of minimum cardinality for~$G$ such that, after adding the shortcuts to $G$, it becomes a $(k,\rho)$-graph.
    \end{itemize}
\end{definition}

To provide a clear characterization in terms of fixed $k$ and $\rho$, we analyze the decision version of the \(\minkrhoShortcutProb\) defined as follows:
\begin{definition}
The~\textnormal{\textsl{$(k,\rho)$\textsc{-Shortcut}}} problem \emph{(\krhoShortcutProb)}:
\begin{itemize}
    \item Input: A weighted graph~$G$ and~$\alpha\in\N$.
    \item Problem: Decide whether there is a shortcut set~$S$ for~$G$ with~$|S| \le \alpha$ such that after adding the shortcuts to $G$, it becomes a $(k,\rho)$-graph.
\end{itemize}
\end{definition}

Observe that the intractability of deciding whether a shortcut set of size $\alpha$ exists directly implies that the optimization version of the problem is \NP-hard.

In order to obtain graphs with small~$k$ and large~$\rho$, reducing the depth of their algorithm while maintaining almost linear work in the number of edges, Blelloch, Gu, Sun, and Tangwongsan~\cite{DBLP:conf/spaa/Blelloch0ST16} provide a heuristic approach.
However, they did not study the complexity of finding the optimal shortcut set.
Subsequent work by Dong, Gu, Sun, and Zhang~\cite{DBLP:conf/spaa/Dong00Z21} also relies on the notion of~$(k,\rho)$-graphs for their further improved algorithm.
Their experimental results underline the practical relevance of finding shortcuts efficiently for arbitrary~$k$ and~$\rho$.
Therefore, characterizing the complexity of~\(\minkrhoShortcutProb\) has emerged as a key theoretical challenge in the design of work-efficient parallel shortest-path algorithms.

Leonhardt, Meyer, and Penschuck~\cite{DBLP:conf/esa/Leonhardt0P24} study the complexity of~$(k,\rho)$-shortcutting, i.e.~the problem of finding a minimum cardinality set of shortcuts that converts a given undirected graph with~$n$ vertices into an undirected~$(k,\rho)$-graph.
For $k\geq3$, they prove that finding such a set of shortcuts is \NP-hard if~$\rho = \Theta(n^\varepsilon)$ holds for some~$\varepsilon > 0$.
In particular, this means that $\rho$ is not a constant but grows as a function of the number of vertices of the input graph, which is not allowed in the more restrictive problem formulation that we consider, where both $k$ and $\rho$ are fixed.
Their construction relies on an intricate reduction from \textsc{Vertex Cover} that requires pitchfork gadgets (star-like structures), and sophisticated blow up lemmas to maintain the proof integrity.
In this paper, we present a significantly simpler and more general reduction from \textsc{Hitting Set}, extending the hardness to previously unexplored regimes and providing a more direct insight into the problem's fundamental complexity.
In addition, we expand the domain of the problem by also considering the directed setting.

\subsection{Our Contribution}
Our work focuses on characterizing the complexity of~\(\minkrhoShortcutProb\) on both, undirected and directed graphs.

By definition, any directed graph with non-negative edge weights is a~$(k,\rho)$-graph when~${\rho\leq k}$.
For~$k=1$,~\(\krhoShortcutProb\) is solvable in polynomial time for all~$\rho$~\cite{DBLP:conf/esa/Leonhardt0P24}.
Therefore, our work focuses on the case where~$\rho>k\geq2$. 

While our primary result focuses on the hardness of shortcutting, it is worth noting that when multiple vertices are at the same distance from a source vertex, the set of ``$\rho$ closest neighbors'' becomes ambiguous.
In such cases, deciding which specific vertices must be reachable using paths of $k$ edges can significantly alter the optimal shortcut set.
This challenge is captured by the~$\krhoTiebreaker$ problem, in which given a graph, the task is  to select a specific set of $\rho$ closest neighbors for each vertex, such that the selection is compatible with an optimal solution to the~\(\minkrhoShortcutProb\) problem. A formal definition of the problem is given in \cref{defn:tiebreaker-formal}.
We show that~$\krhoTiebreaker$ is a hard problem on its own by providing a straightforward reduction from \textsc{Hitting Set}.
In particular, we show the following result in \cref{sec:classifying-k-rho-msp}:
\begin{restatable}[Tiebreaker hardness]{theorem}{tieBreakerhardness}\label{corollary:tiebreaker-hardness}
    For all~$k\geq2,\rho\geq k+3$,~$\dHittingSet{d}$ is polynomial-time Turing-reducible to~$\krhoTiebreaker$ where \(d=\lfloor(\rho-k+1)/2\rfloor\).
    In particular, the  $\krhoTiebreaker$ problem is \NP-hard under Turing reductions for~$k\geq2$ and~$\rho\geq k+3$.
    This remains true when restricting the problem to only undirected graphs or to only directed graphs.
\end{restatable}

This raises the question of whether restricting to instances where the $\rho$-closest neighbors are unique makes the problem more tractable.  
We show \NP-completeness for both undirected and directed graphs by providing a simple reduction from \textsc{Hitting Set}.
Concretely, we show the following in \cref{sec:classifying-k-rho-msp}:
\begin{restatable}[NP-hardness]{mtheorem}{NPhardness}\label{thm:hitting-reduction}
    For all~$k\geq2,\rho\geq k+2$,~$\dHittingSet{d}\leq_p\krhoShortcutProb$ where \(d=\rho-k\).
    Thus,~$\krhoShortcutProb$ is \NP-complete for~$k\geq2$ and~$\rho\geq k+2$.
    This remains true when restricting the problem to only undirected graphs or to only directed graphs where the $\rho$-closest neighbors are unique.
\end{restatable}

This shows that the hardness is not a result of selecting the closest neighbors but it is inherent to the structural optimization of the shortcuts.
In the undirected case, our result significantly improves on the hardness bound of Leonhardt, Meyer, and Penschuck~\cite{DBLP:conf/esa/Leonhardt0P24} as we do not require instances where~\(\rho\) grows as a function of the number of vertices in order to exhibit hardness.
Furthermore, we extend the \NP-hardness result to the directed setting, which, to the best of our knowledge, has not been explored previously.

This leaves only the case where~$\rho=k+1$.
We use the symmetry of the undirected case to solve $\text{min-}(k,k+1)\text{-}\mathrm{Shortcut}$ in polynomial-time in \cref{sec:k_k+1}.
Hence, we fully resolve the complexity of~\(\minkrhoShortcutProb\) on undirected graphs.
Concretely, we show:
\begin{restatable}{mtheorem}{kplusoneundirected}
	The {\minkrhoShortcutProb} problem is polynomial-time solvable on undirected graphs with strictly positive weights.
\end{restatable}

Therefore, the sole open case is $\rho=k+1$ in the directed setting, which is briefly discussed in \cref{sec:conclusion}.

\subsection{Related Work}

Hub labeling is a related field of research.
Cohen, Halperin, Kaplan, and Zwick~\cite{DBLP:journals/siamcomp/CohenHKZ03} propose a data structure for representing information about every pair of vertices in a graph that relies on labeling vertices, which are then used for answering queries quickly and efficiently.
For example, the labels encode the (approximate) distance or the existence of a path between each pair of vertices.
They introduce \emph{2-hop labels} in different variants depending on the information that is represented, and show that these problems can be solved by finding a \emph{2-hop cover}.
This is a collection~$H$ of (certain) paths in the graph such that for every pair of vertices~$u \neq v$, a path from~$u$ to~$v$ can be retrieved using at most two concatenated paths from~$H$.
If there is no restriction on~$H$, then 2-hop labels for answering reachability queries are obtained.
For obtaining 2-hop labels for answering distance queries,~$H$ is restricted to contain only \emph{shortest paths}.
In this case, using our notion of distance-preserving \emph{shortcuts},~$H$ is essentially a set of shortcuts such that within two hops, each vertex can reach every other reachable vertex.
As such, they consider a special variant of the~$(k,\rho)$-shortcutting problem where~$k = 2$ and~$\rho$ is (e.g.) the number of vertices.
For this problem, there are \NP-hardness results as well as bounds on the approximation factor if the restriction for the paths in~$H$ has certain properties.
There is also research considering the case where $k\geq 3$.
For a more elaborate variant of the problem where vertices have a rank and potential hubs are further restricted by these ranks, see~\cite{DBLP:conf/mfcs/BabenkoGKSW15}.

\section{Preliminaries}\label{sec:preliminaries}
We proceed with a formal definition of the problem.
Let~$G$ denote a weighted graph with vertex set~$V(G)$, edge set~$E(G)$, and edge weights~$w_e \ge 1$ for all~$e \in E(G)$.
$G$ can be either directed or undirected.
For any vertices~$u, v \in V(G)$, a \emph{shortest path from~$u$ to~$v$} is a path from~$u$ to~$v$ with minimum total weight.

\begin{definition}[\text{Weight Distance}, \text{Hop Distance}]
	Let~$u, v \in V(G)$.
	The \emph{weight distance} from~$u$ to~$v$ is the total weight of a shortest path from~$u$ to~$v$, and it is denoted by~$\dist(u,v)$.
	The \emph{hop distance} from~$u$ to~$v$ is the minimum number of edges on a shortest path from~$u$ to~$v$, and it is denoted by~$\hopdist(u,v)$.
	If there is no path from~$u$ to~$v$ in~$G$, we set~${\dist(u,v) = \hopdist(u,v) = \infty}$.
\end{definition}

For every vertex~$u \in V(G)$, let~$R_u \coloneqq \{v \in V(G) \setminus \{u\}: \dist(u,v) < \infty \}$ denote the set of vertices that are reachable from~$u$, except $u$ itself.

\begin{definition}[$\rho$-closest Neighbor Set]\label{defn:rho-closest-neighbor-set}
	Let~$u \in V(G)$ and~$\rho \in \N$.
	A set~$S^u_\rho \subseteq R_u$ is called a~\emph{$\rho$-closest neighbor set of vertex~$u$} if and only if it satisfies the following conditions:
    \begin{enumerate}
        \item~$|S^u_\rho|=\min(|R_u|,\;\rho)$.
        \item For all~$v \in S^u_\rho$,~$w \in R_u \setminus S^u_\rho$, it holds that~$\dist(u,v) \leq \dist(u, w)$.
    \end{enumerate}
\end{definition}

Observe that if no more than~$\rho$ vertices are reachable from~$u$, then any set~$S^u_\rho$ needs to contain exactly all of those vertices, which implies~$S^u_\rho = R_u$.
Otherwise,~$S^u_\rho$ is a (not necessarily unique) set of exactly~$\rho$ vertices that have the smallest weight distances from~$u$.

With this definition, we state a local criterion certifying the adherence of a single vertex to the~$(k,\rho)$-property for any~$k, \rho \in \N$.

\begin{definition}[\text{$(k,\rho)$-Ball,~$(k,\rho)$-Graph}]
	A vertex~$u \in V(G)$ has a \emph{$(k,\rho)$-ball} if and only if there is a~$\rho$-closest neighbor set~$S^u_\rho$ such that~$\hopdist(u,v) \le k$ holds for all~$v \in S^u_\rho$.
	A (multi)-graph is a~$(k,\rho)$-graph if and only if every vertex of the (multi)-graph has a~$(k,\rho)$-ball.
\end{definition}

Any given graph~$G$ can be transformed into a~$(k, \rho)$-graph for any~$k, \rho \in \N$ by adding new edges that preserve weight distances while reducing at least one hop distance.

\begin{definition}[\text{Shortcut}, \text{$(k,\rho)$-Shortcut Set}]
    \label{def:shortcut}
	Let~$u, v \in V(G)$ be vertices with weight distance~${\dist(u,v) < \infty}$ and hop distance~$\hopdist(u,v) > 1$.
    \begin{itemize}
    \item If $G$ is undirected, a (possibly new) edge $\{u,v\}$ with weight~$\dist(u,v)$ is a \emph{shortcut in~$G$}.
    \item If $G$ is directed, a (possibly new) edge $(u,v)$ with weight~$\dist(u,v)$ is a \emph{shortcut in~$G$}.
    \end{itemize}
    A set~$S$ of shortcuts in~$G$ is a \emph{$(k,\rho)$-shortcut set for~$G$} if and only if~$G$ becomes a~$(k,\rho)$-graph with identical weight distances when the shortcuts from~$S$ are added to~$G$.
\end{definition}
\begin{remark}
	\label{lem:shortcuts-preserve-weight-dist}
	Adding shortcuts preserves the weight distance between any pair of vertices.
\end{remark}

Note that when adding a shortcut set, a multigraph might emerge.
However, any edge parallel to a shortcut must have strictly greater weight than the shortcut itself as otherwise the hop distance of the vertices would be 1.
In particular,~\cref{lem:shortcuts-preserve-weight-dist} implies that the parallel edge was not used in any shortest path before. 
Hence, we can remove such edges from the graph to avoid the technicality of dealing with multigraphs.

Due to non-negative edge weights, every graph~$G$ is already a~$(k,\rho)$-graph if~$\rho \le k$, and the corresponding problems have trivial solutions.
Another observation can be made for~$k = 1$, where~$\min(1,\rho)$-\textsc{Shortcut} is polynomial-time solvable for all~$\rho \in \N$:
Compute a~$\rho$-closest neighbor set~$S^u_\rho$ for every vertex~$u \in V(G)$, where ties are first broken in favor of shorter hop distances, and second arbitrarily.
As~$u$ needs to have a~$(1,\rho)$-ball, it is sufficient and necessary to insert a shortcut from~$u$ to~$v$ for each vertex~$v \in S^u_\rho$ with~$\hopdist(u,v) > 1$.
Thus, a~$(1,\rho)$-shortcut set that has minimum cardinality can be computed in polynomial time.
Consequently, the corresponding decision problems are also polynomial-time solvable.
Hence, we focus on those problems where~$\rho > k > 1$ holds.

\section{Hardness of the \texorpdfstring{\boldmath$\krhoShortcutProb$}{(k,ρ)-Shortcut} Problem}\label{sec:classifying-k-rho-msp}
In this section, we prove \cref{thm:hitting-reduction,corollary:tiebreaker-hardness} by giving a simple reduction from~\dHittingSet{d}.
Given a hypergraph~\(H\) with vertex set~$V(H)$ and hyperedge set~$E(H)$, a \emph{hitting set} is a subset~\({U\subseteq V(H)}\) such that~\(U\cap e\neq\emptyset\) for all~\(e\in E(H)\).
This is a generalization of the vertex cover problem.
Given a hypergraph~\(H\) and a natural number~\(\alpha\in\N\), the~\HittingSet\ problem asks if~\(H\) contains a hitting set of size at most~$\alpha$.
This is one of Karp's classical \NP-complete problems \cite{DBLP:conf/coco/Karp72}.
The size of the largest hyperedge in~\(H\) is called the rank of~\(H\).
By restricting the rank of the input hypergraph to be at most~\(d\), one obtains the problem~\dHittingSet{d}.
For \(d\geq2\), this restriction remains \NP-complete.

We begin by proving hardness of $\krhoShortcutProb$.
Observe that $\krhoShortcutProb$ can be seen as two nested covering problems:
First, a $\rho$-closest neighbor sets must be selected for each vertex.
Then, shortcuts must be selected so that every vertex reaches its selected $\rho$-closest neighbor set within $k$ hops.
Interestingly, the $\rho$-closest neighbor sets are unique in our reduction.
This motivates us to investigate $\krhoTiebreaker$ afterwards to determine if the first covering problem is also hard on its own.

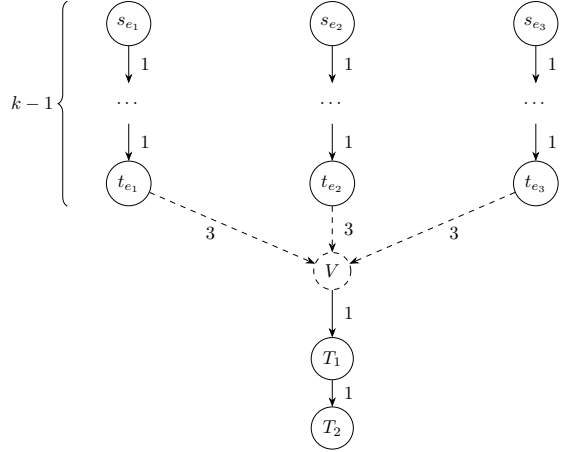
\begin{figure}[!t]
	\centering
	\scalebox{0.7}{
		\begin{tikzpicture}[every node/.style={draw,circle}]
			\foreach \i in {1,2,3}{
				\pgfmathtruncatemacro{\oldi}{\i-1};
				\ifnum\oldi=0
				\node (v\i) {$s_{e_\i}$};
				\else
				\ifnum\i=2
				\node[right=3cm of v\oldi] (v\i) {$s_{e_\i}$};
				\else
				\node[right=3cm of v\oldi] (v\i) {$s_{e_\i}$};
				\fi
				\fi
				\ifnum\i=1
				\def\colorarrow{orange}
				\else
				\ifnum\i=2
				\def\colorarrow{red}
				\else
				\def\colorarrow{blue}
				\fi
				\fi
				
				\foreach \j in {1,2}{
					\pgfmathtruncatemacro{\oldj}{\j-1};
					\ifnum\j=2
					\node[below=7mm of vbl\i\oldj] (vbl\i\j) {$t_{e_\i}$};
					\ifnum\i=2
					\node[below=55mm of v\i] (vbm\i\j) {$T_1$};
					\node[below=5mm of vbm\i\j] (T2) {$T_2$};
					\draw[-Stealth] (vbm\i\j) edge node[right,draw=none] {1} (T2);
					\fi
					\else
					\ifnum\j=1
					\def\belowdistl{7mm}
					\node[draw=none, below =\belowdistl of v\i] (vbl\i\j) {\ldots};
					\else
					\node[draw=none,below=of vbl\i\oldj] (vbl\i\j) {\ldots};
					\fi
					\fi
					\ifnum\j=1
					\draw[-Stealth] (v\i) edge node[right,draw=none] {1} (vbl\i\j);
					\else
					\draw[-Stealth] (vbl\i\oldj) edge node[right,draw=none] {1} (vbl\i\j);
					\fi
				}
			}
			\node[dashed] (V) at ($(vbl22)!0.5!(vbm22)$) {$V$};
			\draw[-Stealth, dashed] (vbl12) edge node[left,xshift=-1mm,yshift=-1mm,draw=none] {3} (V);
			\draw[-Stealth, dashed] (vbl22) edge node[right,draw=none] {3} (V);
			\draw[-Stealth, dashed] (vbl32) edge node[right,xshift=1mm,yshift=-1mm,draw=none] {3} (V);
			\draw[-Stealth] (V) edge node[right,draw=none] {1} (vbm22);
			\useasboundingbox (current bounding box);
			\draw[decorate,decoration={brace,raise=33pt,amplitude=5pt,mirror}] (v1.north) to node[draw=none,rectangle,left=1.3cm] {$k-1$} (vbl12.south);
	\end{tikzpicture}}
	\caption{Exemplary construction of the graph~$G$ from the proof of \cref{thm:hitting-reduction} in the directed setting for a hypergraph~$H$ with three hyperedges~$e_1,e_2,e_3$.
		The dashed lines represent the edges from~$t_{e_i}$ to those~$v\in V(H)$ with~$v\in e_i$.}
	\label{fig:reduction}
\end{figure}

\NPhardness*
\begin{proof}
	Let~$k\geq2,\rho\geq k+2$.
	First, we argue that~$\krhoShortcutProb \in \NP$ using the witness characterization.
	It can be verified in polynomial time if any given set~$S$ is an appropriate~$(k,\rho)$-shortcut set for $G$:
	First, ensure that~$S$ contains only shortcuts for~$G$ and that~$|S| \le \alpha$ holds.
	If this is true, add the shortcuts from~$S$ to~$G$, and verify that~$G$ has become a~$(k,\rho)$-graph by checking that each vertex~$u$ has a~$(k,\rho)$-ball.
	This is achieved by computing the distances between all pairs of vertices in~$G$ and then determining the smallest integer~$d_u$ such that the set~$B_u(d_u)$ contains at least~$\min\{|R_u|,\rho\}$ elements for each~$u\in V(G)$, where~$B_u(r)\coloneqq\{v\in V(G): \dist(u,v)\leq r\}$.
	It only remains to check that the added shortcuts allow each~$u$ to reach all elements of~$B_u(d_u-1)$ and at least~$\min\{|R_u|,\rho\}$ elements of~$B_u(d_u)$ within~$k$ hops.
	This is done by running an augmented version of Dijkstra's algorithm from each vertex~$u$, such that, in addition to the distance, the number of hops on the tentative shortest path is also maintained.
	Whenever a new shortest path of equal distance is discovered, ties are broken by fewer number of hops.
	
	For proving \NP-hardness, fix~$k\geq2,\rho\geq k+2$ and set \(d=\rho-k\).
	Let~$(H,\alpha)$ be an instance of~$\dHittingSet{d}$.
	We begin by giving a reductions which maps only to directed graphs.
	Preprocess the hypergraph \(H\) in the following way:
	For each hyperedge \(e\in E(H)\) with \(|e|<d\), add new unique vertices \(v_1^e,\dots,v_{d-|e|}^e\) and replace~\(e\) with the hyperedge \(e\cup\{v_1^e,\dots,v_{d-|e|}^e\}\).
	This padding ensures that each hyperedge has size exactly~\(d\).
	Add an isolated vertex to \(H\) to ensure that $|V(H)|\geq d+1$ if $E(H)\neq\emptyset$.
    Clearly, the new vertices do not influence the existence of a hitting set of size \(\alpha\).
	
	Unless specified otherwise, the edges have unit weights.
	For~$j\in\N$, denote by~$P_j$ a directed path consisting of~$j$ vertices.
	Construct a graph~$G$ as follows:
	Introduce a path~$P^e\cong P_{k-1}$ for every hyperedge~$e\in E(H)$, denote its starting vertex by~$s_e$ and its final vertex by~$t_e$.
	For every~$v\in V(H)$, add a vertex~$v$ and add the edge~$(t_e,v)$ with weight~3 if and only if~$v\in e$.
	Finally, add vertices~$T_1$ and~$T_2$, the edge~$(T_1,T_2)$ and the edges~$(v,T_1)$ for all~$v\in V(H)$.
	The input~$(H,\alpha)$ is mapped to~$(G,\alpha)$.
	This is clearly computable in polynomial time.
	The construction is visualized in \cref{fig:reduction}.
	
	We now show correctness of the reduction.
	Note that the only vertices that do not have a~$(k,\rho)$-ball are the vertices~$s_e$ for all~$e\in E$.
	These vertices reach~$k+|e|-1=\rho-1$ vertices in~$k$ hops, but can only reach~$T_2$ in~$k+1$ hops.
	Therefore, every~\(s_e\) requires some shortcut to reach~\(T_2\) thereby achieving the~$(k,\rho)$-condition.
	
	$\Rightarrow$: Let~$U$ be a hitting set of size at most~$\alpha$.
	The shortcuts~$\{(v,T_2):v\in U\}$ turn~$G$ into a~$(k,\rho)$-graph:
	Every edge of~$H$ is hit, so every vertex~$s_e$ of~$G$ can use one of the shortcuts~$(v,T_2)$ to reach~$T_2$ in~$k$ hops.
	
	$\Leftarrow$: Let~$S$ be a shortcut set of size at most~$\alpha$ turning~$G$ into a~$(k,\rho)$-graph.
	Assume there is some shortcut~$(u,w)\in S$ such that~$u\notin V(H)$.
	No shortcut can begin in~$T_1$ or~$T_2$, so~$u$ lies in~$V(P^e)$ for some~$e\in E(H)$.
	Then the only vertex whose~$(k,\rho)$-ball depends on~$(u,w)$ is~$s_e$.
	Let~$v\in e$ be arbitrary.
	The set~$S'=(S\setminus\{(u,v)\})\cup\{(v,T_2)\}$ is also a shortcut set and its size satisfies~$|S'|\leq|S|$.
	Thus, we may assume that all edges in~$S$ are of the form~$(v,T_2)$ for some~$v\in V(H)$.
	Every vertex~$s_e\in V(G)$ requires some shortcut to reach~$T_2$ in~$k$ hops, so the set~$U=\{v\in V(H): (v,T_2)\in S\}$ is a hitting set of size at most~$\alpha$.
	
	Next, we give a reduction that maps only to undirected graphs.
    This reduction is obtained by replacing every directed edge in the previous reduction with an undirected edge with the same endpoints and weight.
	In the resulting graph $G$, every vertex reaches every other vertex within~2 hops if \(E(H)=\emptyset\).
	Otherwise, the vertices~$s_e$ for all~$e\in E$ require a shortcut to reach $T_2$ within $k$ hops as in the directed case.
	This is the unique closest neighbor as all other vertices that are not already reachable with $k$ hops require two edges of weight~3 to reach.
	For every hyperedge \(e\in E(H)\), every vertex in \(V(P^e)\cup e\cup\{T_1\}\) except for \(s_e\) can reach every vertex in \(V(P^e)\cup e\cup\{T_1,T_2\}\) within \(k\) hops, which are already~\(k-1+d+2\geq\rho\) vertices.
	The vertex \(T_2\) can reach all but the first vertex of every path \(P^e\) as well as~\(V(H)\) and~\(T_1\) within \(k\) hops.
	As \(|V(H)|\geq d+1\), this implies that \(T_2\) reaches at least \(k-2+d+1+1=\rho\) vertices within \(k\) hops.
	Again, every~\(s_e\) requires some shortcut to reach~\(T_2\) in $K$ hops.
	
	$\Rightarrow$: Let~$U$ be a hitting set of size at most~$\alpha$.
	The shortcuts~$\{\{v,T_2\}:v\in U\}$ turn~$G$ into a~$(k,\rho)$-graph:
	Every edge of~$H$ is hit, so every vertex~$s_e$ of~$G$ can use one of the shortcuts~$(v,T_2)$ to reach~$T_2$ in~$k$ hops.
	
	$\Leftarrow$: Let~$S$ be a shortcut set of size at most~$\alpha$ turning~$G$ into a~$(k,\rho)$-graph.
	Similar to before, we may assume that every shortcut is of the form \(\{v,T_2\}\) for some $v\in V(H)$.
	Then the set~$U=\{v\in V(G): \{v,T_2\}\in S\}$ is a hitting set of size at most~$\alpha$ as every vertex~$s_e\in V(G)$ requires some shortcut to reach~$T_2$ in~$k$ hops.
\end{proof}

We already argued that~$\krhoShortcutProb$ is polynomial-time solvable in the cases of $k = 1$ or $\rho \le k$ in~\cref{sec:intro}.
Thus,~\cref{thm:hitting-reduction} covers all remaining cases except for~$\rho=k+1$ with~$k\geq 2$.
This makes it seem unlikely that~$\minkrhoShortcutProb$ can be solved in polynomial-time for any practically relevant values of~$k$ and~$\rho$.
We show that the case $\rho=k+1$ is polynomial-time solvable for undirected graphs with strictly positive weights in \cref{sec:k_k+1}.

The proof of \cref{thm:hitting-reduction} also yields a lower bound for the approximation factor of a variant of~$\minkrhoShortcutProb$ where only~$k$ is a fixed constant, but~$\rho$ is part of the input.
In the following, {\textsc{Min}-\HittingSet} denotes the canonical optimization variant of \HittingSet.
For~$k\in\N$, $\minkShortcutProb$ is the optimization problem of determining the minimum number of shortcuts that turn a given graph \(G\) into a \((k,\rho)\)-graph, where \(\rho\) is part of the input.

\begin{corollary}\label{corollary:approx-hardness}
	For~$k\geq2$, there is no approximation algorithm for~$\minkShortcutProb$ with approximation factor better than~$|V(G)|$ unless $\P=\NP$.
\end{corollary}
\begin{proof}
	The reductions from~$\dHittingSet{d}$ to~$\krhoShortcutProb$ with~${d=\rho-k}$ can be seen as a reduction from~$\text{min}\HittingSet$ to~$\minkShortcutProb$ with fixed~$k\geq2$, by first determining the rank $d$ of the input hypergraph and then choosing~$\rho$ to be~$d+2$.
	This is uniformly possible as the reductions only depend on~$d$ in the padding step.
	Thus, it transfers the hardness of approximation of~{\textsc{Min}-\HittingSet} shown in \cite{DBLP:journals/jacm/Feige98} to~$\minkShortcutProb$.
\end{proof}

In our reduction, the~$\rho$-closest neighbor sets are unique.
Therefore, the hardness of the problem does not stem from the initial selection of which vertices should be in the~$(k,\rho)$-balls, since there is no decision involved.
However, even deciding which~$\rho$-closest neighbor sets lead to an optimal solution for~$\minkrhoShortcutProb$ is generally hard, as we will show in the following.
First, we restate the $\krhoTiebreaker$ problem more formally:
\begin{definition}\label{defn:tiebreaker-formal}
	The~$\krhoTiebreaker$ problem:
	\begin{itemize}
		\item Input: A weighted graph~$G$.
		\item Problem:
		Find a~$\rho$-closest neighbor set~$S^u_\rho$ for each vertex~$u \in V(G)$ such that the following condition is fulfilled:
		There is an optimal solution for~$\minkrhoShortcutProb$ such that for each~$u\in V(G)$ in the emerging graph (with added shortcuts), it holds~$\hopdist(u,v) \le k$ for all~$v \in S^u_\rho$.
	\end{itemize}
\end{definition}

In other words,~$\krhoTiebreaker$ asks for a~$\rho$-closest neighbor set for each vertex of the graph such that there is an optimal shortcut set where the respective neighbor sets constitute~$(k,\rho)$-balls.
Observe that for each vertex~$v$ there can be multiple eligible vertices for a~$\rho$-closest neighbor set if there are at least two vertices at the boundary of the~$(k,\rho)$-ball with the same distance to~$v$.

\tieBreakerhardness*

\begin{proof}
	Let~$k\geq2$ and~$\rho\geq k+3$, set \(d=\lfloor(\rho-k+1)/2\rfloor\).
	Again, we give an undirected and a directed reduction, where the undirected reduction is obtained from the directed one by replacing directed edges with undirected edges with the same endpoints and weight.
	The directed construction from the proof of \cref{thm:hitting-reduction} is easily adjusted to encode the solution in the choice of the~$\rho$-closest neighbor sets.
	Let~$(H,\alpha)$ be an instance of~$\dHittingSet{d}$.
	Preprocess the hypergraph~\(H\) in the following way:
	For each hyperedge \(e\in E(H)\) with \(|e|<d\), add new unique vertices \(v_1^e,\dots,v_{d-|e|}^e\) and replace \(e\) with the hyperedge \(e\cup\{v_1^e,\dots,v_{d-|e|}^e\}\).
	This padding ensures that each hyperedge has size exactly~\(d\).
	Add an isolated vertex to \(H\) to ensure that $|V(H)|\geq d+1$ if $E(H)\neq\emptyset$.
    Clearly, the new vertices do not influence the existence of a hitting set of size \(\alpha\).
	
	Unless specified otherwise, the edges have unit weights.
	For~$j\in\N$, denote by~$P_j$ a directed path consisting of~$j$ vertices.
	Construct a graph~$G$ as follows:
	Introduce a path~$P^e\cong P_{k-1}$ for every hyperedge~$e\in E(H)$, denote its starting vertex by~$s_e$ and its final vertex by~$t_e$.
	For every~$v\in V(H)$, add the vertices~$v_1,v_2,v_3$  and add the edges~$(t_e,v_1)$ with weight~3 if and only if~$v\in e$.
	Additionally, add the edges~$(v_1,v_2)$ and $(v_2,v_3)$ for every \(v\in V(H)\).
	If $\rho-k+1$ is odd, add a new vertex~$w$ as well as the edge~$(t_e,w)$ with weight~3 for all~$e\in E(H)$.
	
	If $\rho-k+1$ is odd, the vertex $w$ ensures that each of the vertices $s_e$ reaches $k-2+2d+1=\rho-1$ of its nearest neighbors in $k$ hops.
	If $\rho-k+1$ is even, each of the vertices $s_e$ reaches~$k-2+2d=\rho-1$ of its nearest neighbors in $k$ hops.
	Each $s_e$ must receive a shortcut that allows it to reach one of the vertices \(v_3\) with \(v\in e\) within $k+1$ hops.
	As before, these are the only vertices that require shortcuts in both the undirected and the directed case.
	
	Similar to the proof of \cref{thm:hitting-reduction}, a hitting set of size~$\alpha$ corresponds to~$\alpha$ shortcuts of the form~$(v_1,v_3)$.
	Given the sets~$S^{s_e}_\rho$ for all~\(e\in E(H)\), it is clear which shortcuts must be added to achieve an optimal solution, namely the edges with endpoints $\{v_1,v_3\}$ for all~$v\in V(H)$ such there is an~$e\in E(H)$ such that~$v_3\in S^{s_e}_\rho$.
	The sets $S^{s_e}_\rho$ can be computed using the oracle for $\krhoTiebreaker$.
	This results in a polynomial-time Turing reduction.
\end{proof}

\section{Structure of the \texorpdfstring{\boldmath$(k,k+1)$}{(k,k+1)}-Setting}\label{sec:k_k+1}
In the following we show that both {\krhoTiebreaker} and {\krhoShortcutProb} are polynomial-time solvable for $\rho=k+1$ on undirected graphs with strictly positive weights, thus addressing the only remaining open case in the undirected setting.

Throughout this section, unless otherwise stated, we consider graphs with strictly positive weights.
We will start with some general definitions not directly targeted at the
$\rho=k+1$ case to ease the exposition of the following lemmas.
Given a graph $G$, $\mathcal{X}\subseteq V(G)$ denotes the set of vertices that do not form a $\krhoBall{\rho}$ in~$G$.
By definition, all $\rho$-closest neighbor sets~$S^v_\rho$ for any vertex
$v\in\mathcal{X}$ contain at least one vertex that is not reachable by a
shortest path with fewest hops originating from $v$ of less than~$k+1$ edges.
A vertex $v\in\mathcal{X}$ has multiple corresponding
$\rho$-closest neighbor sets $S^v_\rho$ whenever $\max_{u\in
S^v_\rho} \dist(u) = \min_{u\in V\setminus S^v_\rho} \dist(u)$.

\begin{definition}[$\rho$-constrained shortest path tree]
	Let $G$ be an undirected graph and $\varphi:V(G)\to\N$ be a strict ordering of its vertices.
    For~$v\in V(G)$, let $S=\bigcup S^v_\rho$ be the union of all $\rho$-closest neighbor sets of~$v$ and~$S^\varphi\subseteq S$ be the subset of $S$ containing exactly the $\rho$ smallest vertices of $S$ according to the lexicographic order induced by the tuple~${(\dist(v,u),\varphi(u))}$, where~$u\in S$.
  The graphs $\rhoconstspt{\rho}{v}{}$ and $\rhoconstspt{\rho}{v}{\varphi}$ are the shortest path trees with fewest hops connecting~$v$ to all other vertices in $S$ and $S^\varphi$, respectively.
\end{definition}

\begin{definition}[$(k,\rho)$-restricted subgraph]
	Let~$G$ be an undirected graph.
	The~$(k,\rho)$-restricted subgraph~$\krhoRestrictedSubgraph{\rho}$ of $G$ has vertex set~$V(G)$ and edge set
	\begin{equation*}
	    E(\krhoRestrictedSubgraph{\rho})= \bigcup_{v\in\mathcal{X}}
	    E(\rhoconstspt{\rho}{v}{}).
	\end{equation*}
	Given a strict ordering of the vertices $\varphi\colon V\to\N$, the graph $\krhoRestrictedSubgraph{\rho}^\varphi$ is the subgraph of $G$ with vertex set $V(G)$ and edge set
	\begin{equation*}
	    E(\krhoRestrictedSubgraph{\rho}^\varphi)=\bigcup_{v\in\mathcal{X}}
	    E(\rhoconstspt{\rho}{v}{\varphi}).
	\end{equation*}
\end{definition}

We will now prove the existence of a polynomial time algorithm for the $\rho=k+1$ case on undirected graphs.
We begin by analyzing the structure of the~$(k,k+1)$-restricted subgraph.

\begin{lemma}\label{lemma:is-path}
	Let $G$ be an undirected graph with strictly positive edge weights and a strict ordering of the vertices~$\varphi$.
    Let~$v\in\mathcal{X}$ for $\rho=k+1$, then the following holds:
    \begin{enumerate}
		\item $\rhoconstspt{k}{v}{}$ is a path.
		\item $\rhoconstspt{k+1}{v}{}$ has height exactly $k+1$, where the root $v$ is at level~$0$.
		\item $\rhoconstspt{k+1}{v}{\varphi}$ is a path.
    \end{enumerate}
\end{lemma}
\begin{proof}
  The tree $\rhoconstspt{k+1}{v}{}$ must have height at least $k+1$ as otherwise $v$ would not be a member of $\mathcal{X}$.
	As the edge weights are positive, this implies that $v$ can reach at least $k$ of its closest neighbors within $k$ hops.
	If $\rhoconstspt{k}{v}{}$ is not a path, a level of
	$\rhoconstspt{k}{v}{}$ must contain an additional closest neighbor.
    But then $v$ reaches $k+1$ of its closest neighbors within $k$ hops contradicting the fact that $v\in\mathcal{X}$.    
    As $\rhoconstspt{k}{v}{}$ is a path and all weights are strictly
    positive, $\rhoconstspt{k+1}{v}{}$ cannot have height $k+2$.
    Because $\rhoconstspt{k}{v}{}$ is a path, all vertices on the final
    level of $\rhoconstspt{k+1}{v}{}$ must have the same distance to $v$ by the second condition of \cref{defn:rho-closest-neighbor-set}.
    In $\rhoconstspt{k+1}{v}{\varphi}$, the final level only contains the first vertex according to the order $\varphi$ resulting in a path.
\end{proof}

Intuitively, the next lemma proves that the relation between any pair of
shortest path trees~$\rhoconstspt{k+1}{v}{\varphi},\rhoconstspt{k+1}{u}{\varphi}$ is so constrained that even the graph $\krhoRestrictedSubgraph{k+1}^\varphi$ induced by the union of their edges cannot build any complex structure.
\begin{lemma} \label{lemma:forest}
    Let $G$ be an undirected graph with strictly positive edge weights and a strict ordering of the vertices~$\varphi$.
    Then $\krhoRestrictedSubgraph{k+1}^\varphi$ is a forest.
\end{lemma}
\begin{proof}
    Assume for contradiction that there exists a cycle $C=(c_0,c_1,\ldots,c_{t-1},c_0)$ in $\krhoRestrictedSubgraph{k+1}^\varphi$.
    To simplify the notation, we view all indices modulo $t$.
    Let $W$ be the maximum weight of an edge in the cycle.
    Let $0\leq i\leq t-1$ be such that $\{c_i,c_{i+1}\}$ has weight $W$ and maximizes $\max\{\varphi(c_i),\varphi(c_{i+1})\}$ among all edges of weight $W$ in the cycle.
    Let $v\in \mathcal{X}$ be such that the edge $\{c_{i},c_{i+1}\}$ is
    contained in $E(\rhoconstspt{k+1}{v}{\varphi})$, i.e.~$v$ added $\{c_{i},c_{i+1}\}$ to $\krhoRestrictedSubgraph{k+1}^\varphi$.
    Without loss of generality, assume that $\varphi(c_i)>\varphi(c_{i+1})$.
    By \cref{lemma:is-path}, $\rhoconstspt{k+1}{v}{\varphi}$ is a path which we subsequently call $P_v$.
    We distinguish two cases:
    
    \textbf{Case (1):} $c_{i+1}$ is the parent of $c_i$ in $P_v$.
    Then $\dist(v,c_{i+2}) \leq \dist(v,c_{i+1})+W = \dist(v,c_i)$ but $c_{i+2}$ has to be visited first as $\varphi(c_i)$ is maximal.
    Thus, $c_{i+2}$ is visited before $c_i$ and in extension
    before $c_{i+1}$ which is visited by $P_v$ immediately
    before $c_i$. Proceed inductively:
    Let~$j$ be such that $c_j$ is visited before~$c_{i+1}$.
    Then $\dist(v,c_j) < \dist(v,c_i)$ due to strictly positive weights.
    Thus, $\dist(v,c_{j+1})\leq \dist(v,c_j)+W < \dist(v,c_{i})+W$ and
    $c_{j+1}$ must have also been visited before $c_{i}$, implying it is
    also visited before $c_{i+1}$.
    This is a contradiction as all vertices of the cycle are visited without traversing the edge $\{c_i,c_{i+1}\}$.
    
    \textbf{Case (2):} $c_i$ is the parent of $c_{i+1}$ in $P_v$.
    Then $\dist(v,c_{i+1}) =
    \dist(v,c_i)+W$ and we have $\dist(v,c_{i-1}) \leq
    \dist(v,c_i)+w(\{c_i,c_{i-1}\})$.
    Assume $w(\{c_{i-1},c_i\})<W$ or~${\varphi(c_{i-1}) < \varphi(c_{i+1})}$.
    Then $c_{i-1}$ must be visited by $P_v$ before $c_{i+1}$.
    Since $P_v$ traverses $c_i$ immediately before~$c_{i+1}$,~$c_{i-1}$ is visited before $c_i$.
    Then we can derive a contradiction by using the same inductive argument
    as in case (1).

    Now assume $w(\{c_{i-1},c_i\})=W$ and $\varphi(c_{i-1}) > \varphi(c_{i+1})$.
    Consider the edge $\{c_{i-1},c_{i}\}$ and let $u\in \mathcal{X}$ be
    such that the edge $\{c_{i-1},c_{i}\}$ is contained in $P_u\coloneqq
    \rhoconstspt{k+1}{u}{\varphi}$.
    Assume that the path~$P_u$ visits~$c_i$ before~$c_{i-1}$.
    Then $\dist(v,c_{i-1}) = \dist(v,c_i)+W = \dist(v,c_i)+w(\{c_i,c_{i+1}\})$.
    By assumption, $\varphi(c_{i-1}) > \varphi(c_{i+1})$ holds, so $c_{i+1}$ must be visited before $c_i$.
    Therefore, this case is symmetric to the first subcase of case (2) and we are done.
    Otherwise,~$c_{i-1}$ is visited before $c_i$ by $P_u$.
    Then the argument of case (1) applies.
\end{proof}

\begin{lemma}\label{lemma:small-shortcuts}
    Given a shortcut set~$\mathcal{S}$ transforming an undirected graph~$G$ into a~$(k,k+1)$-graph, we can construct an equivalent shortcut set~$\mathcal{S}'$ such that $\hopdist(u,v) = 2$ for all shortcuts~${\{u,v\} \in \mathcal{S}'}$ and~$|\mathcal{S}'|\leq|\mathcal{S}|$.
    For directed graphs, the analogous statement holds.
\end{lemma}
\begin{proof}
    Let~$\{u,v\} \in \mathcal{S}$ with~$\hopdist(u,v)>2$.
    Then~$\{u,v\}$ lies on some shortest path with fewest hops~$P_{u,v}=(p_0=u,p_1,p_2,\ldots,p_i=v)$ in~$G$.
    Assume for contradiction that replacing the shortcut~$\{u,v\}$
    with~$\{p_{i-2},v\}$ prevents a formerly present~$(k,k+1)$-ball for
    some~$z \in V$ because some vertex~$t \in V$ becomes unreachable from~$z$ within~$k$ hops.
    Since $t$ is the $(k+1)$-th vertex in the shortest path with fewest hops of $z$, we have $\hopdist_{G}(z,t)=k+1$.
    As $\{u,v\}$ is used by $z$, the shortest path with fewest hops of $z$ contains $P_{u,v}$ as a subpath.
    This implies
    \begin{align*}
			k+1&=\hopdist_{G}(z,t)=\hopdist_{G}(z,u)+\hopdist_{G}(u,v)+\hopdist_G(v,t)\\
			 &>\hopdist_{G}(z,u)+\hopdist_{G\cup\{\{p_{i-2},v\}\}}(u,v)+\hopdist_{G}(v,t).
    \end{align*}
    But then $\hopdist_{G\cup\{\{p_{i-2},v\}\}}(z,t) \leq k$, and $z$ can reach $t$ with less than~$k+1$ hops in~${G\cup\{\{p_{i-2},v\}\}}$, a contradiction.
    Thus, any shortcut~$\{u,v\} \in \mathcal{S}$ can be replaced by a shortcut that skips exactly one vertex.
\end{proof}

\begin{definition}[Tiebreaking candidates]
    For any vertex $v\in\mathcal{X}$, let
    \begin{equation*}
      B(v) = \left\{u\,|\, u\in V(\rhoconstspt{k+1}{v}{})\setminus
      V(\rhoconstspt{k}{v}{})\right\}
    \end{equation*}
    be the set of vertices that are present in the $(k+1)$-constrained shortest path tree of $v$ but not in the $k$-constrained shortest path tree.
\end{definition}
Observe, that $\rhoconstspt{k+1}{v}{\varphi}$ contains exactly one vertex $u\in B(v)$ which is decided upon based on the strict ordering $\varphi$.

\begin{lemma}\label{lemma:tiebreaking}
	Let $G$ be an undirected graph and let $u,x\in V(G)$ with $u\neq x$.
	Let $\mathcal{P}$ be the subset of vertices $v\in\mathcal{X}$
	satisfying that the path $\rhoconstspt{k}{v}{}$ has $x$ and $u$ as second last and last vertex respectively.
	Then for any pair of distinct vertices $i,j\in\mathcal{P}$ it holds that $B(i)=B(j)$.
\end{lemma}
\begin{proof}
    Assume for contradiction that there exists pair of vertices $i,j\in\mathcal{P}$ such that $B(i)\setminus B(j)$ is non-empty.
    Let $y\in B(i)\setminus B(j)$.
    By \cref{lemma:is-path}, $\rhoconstspt{k+1}{i}{}$ has height $k+1$, so
    $y$ is in the~$(k+1)$-th level of $\rhoconstspt{k+1}{i}{}$ and
    $\dist(i,y) \leq \min_{t\in V\setminus V(\rhoconstspt{k}{v}{})}\dist(i,t)$.
    
    \textbf{Case 1:} $y\notin B(j)$ holds because there exists a vertex $z\in B(j)$ with ${\dist(j,z) < \dist(j,y)}$.
    Since $j\in\mathcal{X}$, $\rhoconstspt{k+1}{j}{}$ is also a tree of height $k+1$.
    By assumption, the last vertex of both~$\rhoconstspt{k}{i}{}$
    and~$\rhoconstspt{k}{j}{}$ is~$u$.
    Hence, we have $\dist(j,z) = \dist(j,u)+w(\{u,z\}) < \dist(j,u)+w(\{u,y\})$.
    But then, $z\in V(\rhoconstspt{k}{i}{})$ as otherwise $y\notin B(i)$ contrary to the assumption.
    In addition, $z$ must be a neighbor of $u$ due to $z\in B(j)$ and~$j\in\mathcal{X}$.
    Thus, we have $\dist(i,u) = \dist(i,z)+\dist(z,u) <
    \dist(i,z)+w(\{z,u\}$ since $u$ was entered through the edge~$\{x,u\}$
    by both $\rhoconstspt{k}{i}{}$ and $\rhoconstspt{k}{j}{}$.
    But then $\dist(j,u)+w(\{z,u\} > \dist(j,u)+\dist(u,z)$ hence $z$
    cannot possibly be at level $k+1$ of~$\rhoconstspt{k+1}{j}{}$ contradicting the assumption that $z\in B(j)$.
    
    \textbf{Case 2:} $y\notin B(j)$ holds because $y\in
    V(\rhoconstspt{k}{j}{})$.
    Due to $y\in B(i)$ and $i\in\mathcal{X}$, $y$ is a neighbor of $u$.
    Hence, $\dist(j,y)+\dist(y,u) < \dist(j,y)+w(\{y,u\})$.
    But then, there is a shorter path from $y$ to $u$ not using the edge
    $\{y,u\}$ contradicting the fact that $\rhoconstspt{k+1}{i}{}$ is shortest path tree of height~$k+1$.
\end{proof}

Let us consider the following algorithm which produces an optimal solution for the {\minkrhoShortcutProb} problem:
Encode each possible shortcut~$\{u,v\}$ as a set~$E_{u,v} = \{x\in\mathcal{X}\,|\,x\textnormal{ forms a $(k,k+1)$-ball in $G\cup\{\{u,v\}\}$}\}$ and find a minimal set-cover over the resulting set-system with universe $\mathcal{X}$.
A minimal set-cover of the universe $\mathcal{X}$ now corresponds to a globally minimal shortcut set.

\begin{theorem} \label{theorem:tiebreak}
	Let $G$ be an undirected graph with strictly positive weights.
	Then any strict ordering of the vertices $\varphi$ optimally solves the $(k,k+1)\text{-}\mathrm{Tiebreaker}$ problem.
\end{theorem}
\begin{proof}
	Suppose $\varphi$ is a strict ordering of the vertices of $G$.
    Consider the graph $\krhoRestrictedSubgraph{k+1}^\varphi$ and let us call the resulting set system for the set-cover instance $S$. We aim to show now that for each set in the set-system of $OPT$, that is an optimal tiebreaker assignment, there is a set in $S$ which contains the set of $OPT$.

    Take any set $E_{u,v} \in OPT$ corresponding to a shortcut $\{u,v\}$. 
    By \cref{lemma:small-shortcuts}, we may assume that $\hopdist(u,v)=2$.
    If there does not exist a vertex $x\in E_{u,v}$ for which $v\in B(x)$ and $|B(x)| \geq 2$ holds simultaneously,
    tiebreaking has no influence on the
    vertices in $E_{u,v}$ as the choice of $v$ is unambiguous, so $E_{u,v}\in S$.
    Thus, assume there exists a vertex $x\in E_{u,v}$ for which $v\in B(x)$ and $|B(x)|\geq 2$ holds.
    Since $\{u,v\}$ shortcuts the path $(u,y,v)$, any vertex $z\in E_{u,v}$
    contains $(u,y,v)$ as a subpath.
    In addition, there exists a vertex $z \in B(x)\setminus\{v\}$ such that
    $\dist(y,v)=\dist(y,z)$, but then for each vertex $t\in E_{u,v}$ it holds
    that~$\rhoconstspt{k}{t}{}$
    ends with the suffix $u,y$ as otherwise $\rhoconstspt{k}{t}{}$ is not a path
    contradicting $t\in\mathcal{X}$.
    Thus, we can apply~\cref{lemma:tiebreaking}, so all vertices in $E_{u,v}$ have the same prospective vertices eligible for tiebreaking.
    By the tiebreaking rule above, all vertices would have chosen the same vertex as tiebreaker.
    Thus, there exists $E'_{u,j} \in S$ such that $E_{u,v} \subseteq E'_{u,j}$.
    Since this holds for any set in $OPT$, the optimal solution for the set system $S$ is no larger than the one for $OPT$.
\end{proof}

\kplusoneundirected*
\begin{proof}
    Take an arbitrary strict ordering of the vertices $\varphi$, e.g., the lexicographic ordering.
    By \cref{theorem:tiebreak}, tiebreaking via $\varphi$ optimally solves the $(k,k+1)\text{-}\mathrm{Tiebreaker}$ problem.
   	By \cref{lemma:forest}, $\krhoRestrictedSubgraph{k+1}^\varphi$ is a forest.
    Thus, the problem can be solved in polynomial-time by applying Courcelle's theorem~\cite{DBLP:journals/iandc/Courcelle90}.
\end{proof}

\section{Conclusion and Future Work}\label{sec:conclusion}
We have given a simple reduction proving \NP-hardness that can be taught in class.
Although simple in nature, it is sufficiently expressive to improve on the prior state of the art by reducing the prior constraint of $k\geq 3$ for hardness down to $k\geq 2$.
In addition, we do not rely on a polynomial dependency of $\rho$ on $n$ to exhibit hardness unlike previus work.

While we successfully settle the complexity for the undirected case, the directed case for $\rho=k+1$ remains open.
Many of the technical lemmas used for the undirected case, e.g., \cref{lemma:forest,lemma:tiebreaking}, crucially rely on the ability to walk a path backwards.
Thus, the proofs inevitably break down when considering directed paths instead.  
Although the structure in this setting is still highly constrained, any proof has to deal with the asymmetry inherent to directed graphs.
Several statements about $\krhoRestrictedSubgraph{k+1}^\varphi$ turn out to be false for its directed analogue.
For example, the directed variant of $\krhoRestrictedSubgraph{k+1}^\varphi$ is neither a forest nor a pseudoforest.
This suggests that the directed case for $\rho = k+1$ may not be a mere technical variant, but could instead belong to a different complexity class, or that at least it requires a fundamentally different approach.
The settlement of this case is a problem for future work.

Our tractability proof for the undirected case with $\rho=k+1$ relies on Courcelle's theorem.
Although this approach allows us to prove that the problem is polynomial-time solvable, the specific structure of the regime provides a clear opportunity for future work to design a more direct and efficient combinatorial algorithm.

\clearpage
\bibliography{main}
\end{document}